\documentclass[sigconf,natbib=true,screen=true,anonymous=false ,review=false]{acmart}

\author{Philipp Hager}
\orcid{0000-0001-5696-9732}
\affiliation{%
    \institution{University of Amsterdam}
    \city{Amsterdam}
    \country{The Netherlands}
}
\email{p.k.hager@uva.nl}

\author{Onno Zoeter}
\orcid{0000-0003-1704-706X}
\affiliation{%
    \institution{Booking.com}
    \city{Amsterdam}
    \country{The Netherlands}
}
\email{onno.zoeter@booking.com}

\author{Maarten de Rijke}
\orcid{0000-0002-1086-0202}
\affiliation{%
    \institution{University of Amsterdam}
    \city{Amsterdam}
    \country{The Netherlands}
}
\email{m.derijke@uva.nl}

\usepackage[inline]{enumitem}
\usepackage{multirow}
\usepackage{acronym}
\usepackage{stfloats}
\usepackage[skip=0pt]{caption}
\usepackage{etoolbox}
\usepackage{dsfont}

\makeatletter 
\newcommand\mynobreakpar{\par\nobreak\@afterheading} 
\makeatother

\usepackage[belowskip=-12pt,aboveskip=0pt]{caption}

\def\reals{\mathbb{R}}

\newcommand{\supp}[1]{\text{supp}\left( {#1} \right)}

\acmConference[CONSEQUENCES '25]{CONSEQUENCES Workshop at RecSys '25}{September 22, 2025}{Prague, Czech Republic}
\acmBooktitle{CONSEQUENCES Workshop at RecSys '25 (ICTIR '25), September 22, 2025, Prague, Czech Republic}

\begin{document}

\renewcommand\footnotetextcopyrightpermission[1]{}

\title[Unidentified and Confounded? Understanding Two-Tower Models for Unbiased Learning to Rank]{Unidentified and Confounded? Understanding Two-Tower Models for Unbiased Learning to Rank (Extended Abstract)}
\titlenote{Based on a full paper "Unidentified and Confounded? Understanding Two-Tower Models for Unbiased Learning to Rank" published at ICTIR'25~\cite{Hager2025TwoTowers}.}




\maketitle

\vspace*{-1mm}
\section{Introduction}

\begin{figure}
    \hspace*{-2em}
    \includegraphics[clip, trim=0mm -18mm 0mm 0mm, width=0.9\linewidth]{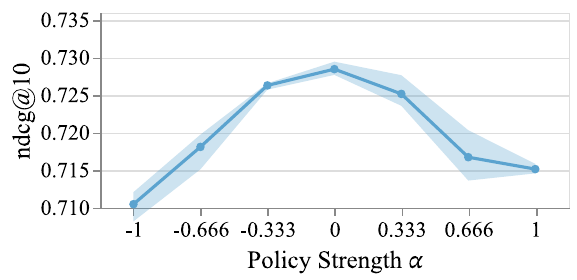}
    \vspace*{-1.4cm}
    \caption{Two-tower models trained on deterministic logging policies of varying strengths ($\alpha$) on MSLR30K: $\alpha = 1$ represents sorting by expert annotations, $\alpha = 0$ random sorting, and $\alpha = -1$ inversely ranking from least to most relevant.}
    \label{fig:example}
\end{figure}

Two-tower models are a popular solution to address click biases in industrial settings~\cite{Guo2019PAL,Yan2022TwoTowers,Zhao2019AdditiveTowers,Haldar2020Airbnb,Khrylchenko2023Yandex}. These models consist of two neural networks, referred to as towers: a relevance tower predicting item relevance based on content-related features and a bias tower predicting user examination from contextual features such as document position or display height on screen~\cite{Yan2022TwoTowers,Chen2023-TENCENT-ULTR-1}. During training, outputs from both towers are combined to predict user clicks, whereas, at serving time, only the relevance tower is employed.

However, recent studies report a troubling phenomenon: as production logging policies become more effective at surfacing relevant items, the ranking performance of two-tower models trained on clicks from those policies steadily declines~\citep{Zhang2023Disentangling}. Two explanations have emerged: \emph{confounding} and \emph{identifiability}.
\citet{Zhang2023Disentangling} theorize that stronger logging policies act as confounders by introducing correlations between bias features and document relevance~\citep{Luo2024UnbiasedPropensity}. While confounding explanations are intuitive, they are theoretically unsatisfactory since two-tower models explicitly condition on both relevance and bias features.\footnote{Two-tower models condition on relevance features and position, essentially already blocking the problematic paths highlighted in \cite[Fig. 2b]{Zhang2023Disentangling} and \cite[Fig. 1b]{Luo2024UnbiasedPropensity}.} Yet, we observe the same empirical effect in reproducibility experiments. Fig.~\ref{fig:example} displays the ranking performance of two-tower models trained on data collected by policies of varying strengths\footnote{Different strengths in Fig.~\ref{fig:example} are based on sorting by expert relevance annotations with additive noise~\cite{Zhang2023Disentangling,Deffayet2023CMIP}. We show that this policy simulation is problematic.} and ranking performance decreases when the logging policy is not uniformly random.

\citet{Chen2024Identifiability} discuss an alternative explanation in terms of identifiability. Intuitively, identifiability asks under which conditions we can (uniquely) recover model parameters from observed data. Inferring model parameters in click modeling is a well-known challenge~\cite{Chuklin2015ClickModels,Joachims2017IPW} as estimating position bias parameters typically requires observations of the same query-document pair across two positions (a swap) to correctly attribute changes in clicks to bias rather than a shift in document relevance~\cite{Craswell2008PBM,Agarwal2019AllPairs,Fang2019InterventionHarvesting}. The influence of logging policies on identifiability seems clear: our production system might not collect suitable data for training two-tower models, e.g., by never swapping documents across positions~\cite{Chen2024Identifiability}. The experiments on logging policy confounding in \cite{Zhang2023Disentangling,Luo2024UnbiasedPropensity} and Fig.~\ref{fig:example} use deterministic logging policies without document swaps, suggesting identifiability might be a culprit.

In this work, we analyze two-tower models by investigating conditions for identifiability and studying the effects of the current production policy on training two-tower models.

\section{Identifiability of Two-Tower Models}
\label{sec:identifiability}
Our analysis focuses on additive two-tower models for position bias correction~\cite{Guo2019PAL,Yan2022TwoTowers,Zhao2019AdditiveTowers,Haldar2020Airbnb}. Each tower outputs logits which are combined using the sigmoid function $\sigma(\cdot)$:
\begin{equation}
    \label{eq:two-towers}
    P(C = 1 \mid q, d, k) = \sigma(\theta_{k} + \gamma_{q,d}),
\end{equation}
where $\theta_{k} \in \reals$ is the position bias logit for rank $k \in \{1, 2, \dots \}$ and $\gamma_{q,d} \in \reals$ is the relevance logit for document $d \in D$ and query $q \in Q$.

We begin our analysis by examining when we can uniquely recover the parameters of two-tower models from click data, fundamentally a question of identifiability. A model is identifiable if no two distinct sets of parameters induce the same distribution over observable variables \cite[Definition 19.4]{Koller2009PGM}. Unidentifiability, in turn, implies that we cannot uniquely recover the true model parameters from observed data as there are multiple (possibly infinite) parameters that explain the observed data equally well.

Two-tower models face inherent identifiability challenges since we only observe clicks while position bias and document relevance remain unobserved, leading to the problem of parameter shifts: we can arbitrarily increase our relevance parameters while decreasing our bias parameters by the same amount, resulting in identical click probabilities but different parameters. Consider a set of alternative parameters $\theta_{k}'$ and $\gamma_{q,d}'$ next to our true parameters $\theta_{k}$ and $\gamma_{q,d}$:
\begin{equation}\label{eq:shift}
    \theta_{k} + \gamma_{q,d} = \theta_{k}' + \gamma_{q,d}' \quad \forall \, (q,d,k).
\end{equation}
We can construct infinite valid alternative parameters: $\theta_{k}' = \theta_{k} + \Delta_k$ and $\gamma_{q,d}' = \gamma_{q,d} - \Delta_k$, where $\Delta_k \in \reals$ can be chosen independently for each rank. When documents appear exclusively in a single position, we can choose an arbitrary offset $\Delta_k$ to shift their relevance parameters and adjust the bias parameter accordingly without changing the overall click probabilities. Therefore, \textit{additive two-tower models are unidentifiable when observing each query-document pair only at a single position~\cite{Chen2024Identifiability,Oosterhuis2022Limitations}}.

\vspace*{-1mm}
\subsection{Identification through randomization}
\label{sec:identifiability-randomization}

When observing query-document pairs across positions, parameter offsets become constraint. If the same query-document pair $q,d$ appears in positions 1 and 2, then the parameter offsets between both positions must be equal: $\gamma_{q,d} - \gamma_{q,d}' = \Delta_1 = \Delta_2.$

Random swaps between positions unify parameter shifts across positions to a single global parameter shift, which is commonly resolved through normalizing parameters~\cite[Section 6.3]{Lewbel2019IdentificationZoo} (e.g., fixing the first bias logit to $\theta_1=0$). Note that not all query-document pairs need to be observed across all positions. \citet{Chen2024Identifiability} observe that it is sufficient that an undirected graph $G = (V, E)$, where vertices $V$ correspond to positions and edges $E$ exist between each pair of ranks that shares at least one query-document pair, is connected~\cite[Theorem 1]{Chen2024Identifiability}. Thus, \textit{additive two-tower models are identifiable up to an additive constant when observing query-document pairs across positions, such that all positions form a connected graph.} 

\vspace*{-1mm}
\subsection{Identification through overlapping features}
\label{sec:Identification-feature-overlap}

Most applications of two-tower models do not use independent model parameters for each query-document pair. Instead, we represent query-document pairs through feature vectors $x_{q,d} \in \reals^m$ and the relevance tower $r(\cdot)$ is, e.g., a linear model or a neural network:
\begin{equation}
    P(C = 1 \mid q, d, k) = \sigma(\theta_{k} + r(x_{q,d})).
\end{equation}
Our main result builds on \cite[Theorem 1]{Chen2024Identifiability} and shows that identifiability is possible without swaps as long as the feature distributions of query-document pairs overlap across positions:
\begin{equation}
    \label{eq:identifiability-features}
        \begin{split}
        \supp{P(x \mid k)} &\cap \supp{P(x \mid k')} \neq \emptyset,
        \end{split}
\end{equation}
and the relevance tower $r(\cdot)$ is continuous. Continuity ensures that small shifts in query-document features do not cause large, discontinuous jumps in relevance prediction. We provide the proof in Appendix~\ref{appendix:overlapping-features} and highlight practical challenges of obtaining overlapping feature distributions in Appendix~\ref{appendix:practical-pitfalls}. In summary, \textit{identifying additive two-tower models requires overlap in distributions across ranks, either through explicit randomization or shared query-document features that allow disentangling bias and relevance}.

\vspace*{-1mm}
\section{Influence of the Logging Policy}
\label{sec:logging_policy}

\begin{figure*}[h!]
    \includegraphics[width=1\textwidth]{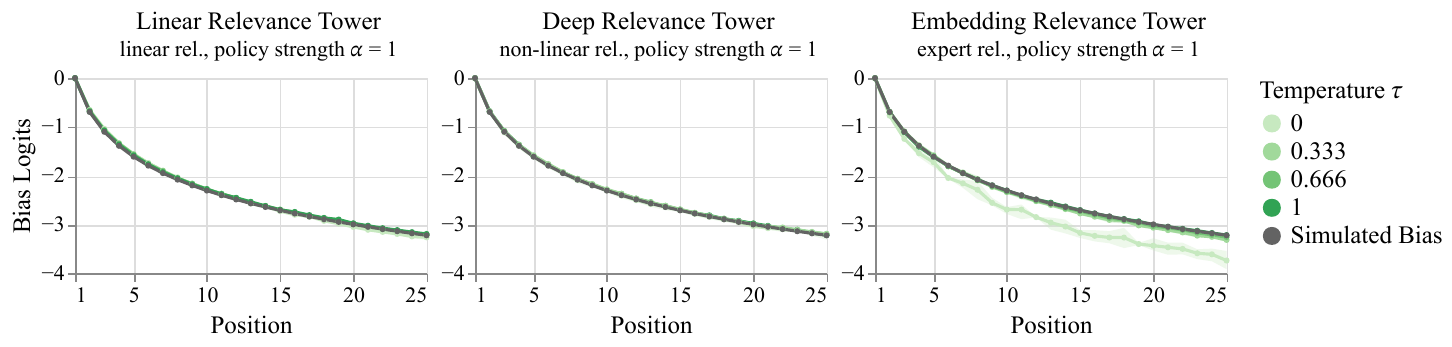}
    \caption{Evaluating position bias for three well-specified models matching the simulated user behavior. All models converge to ground-truth bias parameters while being trained on the strongest logging policy $(\alpha = 1)$. Feature-based models recover position bias without any document swaps ($\tau = 0$), while the embedding-based model requires swaps for identification ($\tau > 0$).}
    \label{fig:model-fit}
\end{figure*}

\begin{figure*}[h!]
    \includegraphics[width=1\textwidth]{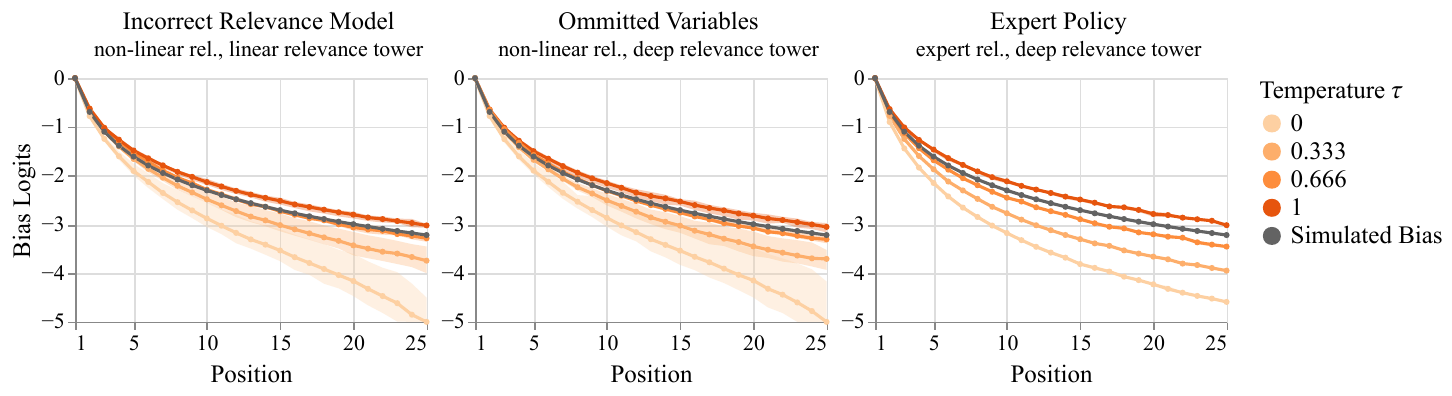}
    \caption{Evaluating position bias under model misspecification: (left) training a linear relevance model on non-linear user behavior, (middle) using fewer features when training two-tower models than are available to the logging policy, and (right) simulating an oracle logging policy by sorting by expert relevance labels from MSLR30K. The logging policy amplifies bias in all models. Model misspecification is visible as estimations do not match the simulated bias on fully randomized data ($\tau = 1$).}
    \label{fig:model-misfit}
\end{figure*}

\vspace*{-1mm}
The connection between a logging policy and identifiability seems clear: we must collect overlapping feature or document distributions across ranks. However, it is unclear if this rather binary condition can explain the gradual degradation in ranking performance as found in related work \cite{Zhang2023Disentangling,Luo2024UnbiasedPropensity} and our experiment in Fig.~\ref{fig:example}. Therefore, we investigate the role of logging policies beyond identifiability. To begin, we highlight the role of the logging policy by rewriting the standard training objective of two-tower models, the expected negative log-likelihood:
\begin{equation}
    \begin{split}
        &\mathcal{L}(\theta, \gamma)=
        - \sum_{q} P(q)\! \sum_{d, k} \pi(d, k \mid q) \bigl[ P(C\!=\!1 \mid q,d,k) \cdot {}
        \\[-4pt]
        &\mbox{}\hspace*{0.5cm} 
        \log \sigma(\theta_k + \gamma_{q,d})+ P(C\!=\!0|q,d,k)\log (1-\sigma(\theta_k + \gamma_{q,d})) \bigr].
    \end{split}
\end{equation}
We define a policy as the joint probability $\pi(d, k \mid q)$ of a document $d$ being shown at rank $k$ for query $q$. First, we can observe that the policy weights the contribution of each query-document pair to the loss. Our analysis below separates the impact of logging policies on well-specified versus misspecified models.

\begin{lemma}[No policy impact on well-specified models]
\label{lemma:well-specified}
If a two-tower model is well-specified (i.e., can perfectly model the true click probabilities) and identifiable, then the logging policy has no effect on the estimated model parameters.
\end{lemma}

\begin{proof}
The partial derivatives of the loss function with respect to the model parameters are:
\begin{equation}
\begin{split}
    &\mbox{}\hspace*{-3mm}\frac{\partial \mathcal{L}}{\partial \gamma_{q,d}} \!=\!{}\\[-3pt]
    &\mbox{}\hspace*{5mm}P(q) \! \sum_{k} \pi(d, k \mid q) \bigl[ P(C\!=\!1 \mid q,d,k) - \sigma(\theta_k + \gamma_{q,d}) \bigr]\!=\!0,
    \hspace*{-3mm}\mbox{}
\end{split}
\label{eq:relevance-derivative}
\end{equation}
\vspace*{-3mm}
\begin{equation}
\begin{split}
    & \frac{\partial \mathcal{L}}{\partial \theta_k} \!=\! {}\\
    & \sum_{q} P(q) \! \sum_{d} \pi(d, k \mid q) \bigl[ P(C\!=\!1 \mid q,d,k) - \sigma(\theta_k + \gamma_{q,d}) \bigr]\!=\!0.
\end{split}    
\label{eq:bias-derivatie}
\end{equation}

For these gradients to vanish at the optimal parameters, the following condition must hold for every query-document pair with non-zero display probability, i.e., when $\pi(d, k \mid q) > 0$:
\begin{equation}
    P(C\!=\!1 \mid q,d,k) = \sigma(\theta_k + \gamma_{q,d}).
\end{equation}
In a well-specified model, this condition can be satisfied for all query-document pairs where $\pi(d, k \mid q) > 0$. Thus, the influence of the logging policy on the estimated parameters vanishes as long as identifiability is guaranteed.
\end{proof}

\begin{lemma}[Logging policy impact on misspecified models]
\label{lemma:misspecified}
When a two-tower model is misspecified, systematic correlations between the model's residual errors and the logging policy can introduce bias in parameter estimation, even when identifiability is guaranteed.
\end{lemma}

\begin{proof}
When the model is unable to match the true click probabilities, we have a non-zero residual click prediction error:
\begin{equation}
    \epsilon(q,d,k) \equiv P(C\!=\!1 \mid q,d,k) - \sigma(\theta_k+\gamma_{q,d}) \neq 0,
\end{equation}
with the gradient conditions becoming:
\begin{equation}\label{eq:gradient-residual}
    \sum_k\pi(d,k\mid q)\,\epsilon(q,d,k)=0,
     \ \
    \sum_q\sum_d\pi(d,k\mid q)\,\epsilon(q,d,k)=0.
\end{equation}
These conditions state that the policy-weighted averages of the residuals must vanish across positions and across items. However, when our model's click prediction errors $\epsilon(q,d,k)$ are systematically correlated with the policy's display probability $\pi(d, k \mid q)$, the optimizer must shift model parameters from their true values to satisfy these conditions. Model misspecification alone can bias model parameters as the average residual error might be non-zero even under a uniform random logging policy, but the logging policy can amplify the bias in parameter estimation.
\end{proof}

\noindent%
We reveal a nuanced relationship between logging policies and additive two-tower models. While perfectly specified models remain unaffected by logging policies (beyond identifiability concerns), real-world models might inevitably contain some degree of misspecification. This creates a vulnerability where systematic correlations between model errors and logging policy behaviors can significantly distort parameter estimation. \emph{The more a logging policy systematically favors certain documents for certain positions, the more these correlations can amplify estimation bias.}

Our analysis reveals logging policies affect two-tower models only through model misspecification. We discuss three critical scenarios:
\begin{enumerate*}[label=(\roman*)]
\item Functional form mismatch: When fitting linear models to non-linear user behavior, prediction errors can correlate with document placement.
\item Omitted variable bias: If logging policies use features unavailable to the current two-tower model (e.g., business rules), systematic prediction errors can emerge that correlate with position~\cite{Wilms2021OmittedVariableBias}.
\item Expert policy in simulation: A common but problematic practice is simulating strong policies by sorting directly on expert labels rather than model predictions~\cite{Zhang2023Disentangling,Deffayet2023CMIP}. This creates artificial omitted variable bias since expert labels contain information not captured by query-document features and explains the problematic experimental results in \cite{Zhang2023Disentangling} and Fig.~\ref{fig:example}.\footnote{Appendix~\ref{appendix:motivating-example} shows the results of our motivating example without sorting by expert annotations.}
\end{enumerate*}

\vspace*{-3mm}
\section{Experiments}
We empirically test our claims by evaluating if well-specified and misspecified two-tower models can recover position bias parameters in simulation experiments (Appendix~\ref{appendix:experimental-setup}). We  simulate the strongest possible logging policy with varying levels of randomization, where $\tau=0$ represents a deterministic ranking and $\tau=1$ represents a uniform random ranking. For well-specified models (Fig.~\ref{fig:model-fit}), we confirm that separate query-document parameters require document swaps for identification (failing under deterministic policies with $\tau = 0$ in the right panel), while feature-based models (left and middle panels) achieve perfect parameter recovery even on deterministic logging policies. In all well-specified models, logging policies have no effect once identifiability is guaranteed. For misspecified models (Fig.~\ref{fig:model-misfit}), we test three scenarios and find that all cases show systematic bias amplification as logging policies become less random (lower $\tau$), confirming that non-uniform logging policy exposure across ranks intensifies bias in misspecified models.

\vspace*{-3mm}
\section{Conclusion}
Our work investigates why two-tower models show degrading performance on data from strong logging policies by examining identifiability conditions and logging policy effects. We show that two-tower models can be identified without document swaps through overlapping feature distributions across ranks, and that logging policies do not affect well-specified models. However, misspecified models with systematic prediction errors experience bias amplification when these errors correlate with document placement, empirically demonstrated across three scenarios. Practically, our work suggests
\begin{enumerate*}[label=(\roman*)]
\item monitoring click residuals for logging policy correlations to detect model misspecification, 
\item collecting randomized data when feasible to ensure overlapping document or feature distributions across positions, and 
\item avoid sorting by expert labels in simulation.
\end{enumerate*}
While practitioners should prioritize reducing model misspecification, when some degree of misspecification is inevitable, we propose a sample weighting scheme to counteract potential logging policy influences in Appendix~\ref{appendix:ips}.

\vspace*{-1mm}
\begin{acks}
We thank Philip Boeken and Shashank Gupta for their insightful discussions and valuable feedback. This research was supported by the Mercury Machine Learning Lab created by TU Delft, the University of Amsterdam, and Booking.com.
Maarten de Rijke was supported by the Dutch Research Council (NWO), project nrs 024.004.022, NWA.1389.20.\-183, and KICH3.LTP.20.006, and the European Union's Horizon Europe program under grant agreement No 101070212. All content represents the opinion of the authors, which their employers or sponsors do not necessarily endorse.
\end{acks}

\bibliographystyle{ACM-Reference-Format}
\balance
\bibliography{main}

\clearpage
\appendix

\section{Identification through overlapping features}
\label{appendix:overlapping-features}

\begin{theorem}[Identifiability through feature overlap]
    \label{theorem:feature-overlap}
    Let $G = (V, E)$ be an undirected graph where vertices $V$ correspond to positions and an edge $(k, k') \in E$ exists if the feature support between positions $k$ and $k'$ overlap:
    \begin{equation}
        \label{eq:identifiability-features}
           \begin{split}
            \supp{P(x \mid k)} &\cap \supp{P(x \mid k')} \neq \emptyset.
           \end{split}
    \end{equation}
    If $G$ is connected and the relevance tower $r(\cdot)$ is continuous, then the additive two-tower model is approximately identifiable up to an additive constant.
    \end{theorem}
    
    \begin{proof}
    Consider the alternative parameterization $\theta'_k$ and $r'(\cdot)$ that yield identical click probabilities to our original parameters $\theta_k$ and $r(\cdot)$. For these parameterizations to produce the same click probabilities, the following must hold:
    \begin{equation}\label{eq:feature_deltas}
        r(x_{q,d}) - r'(x_{q,d}) = \theta'_{k} - \theta_{k} = \Delta_k,
    \end{equation}
    where $\Delta_k$ is a rank-dependent offset. For any two positions $k$ and $k'$ that share overlapping feature support, there exist query-document pairs with feature vectors $x_1 \sim P(x|k)$ and $x_2 \sim P(x|k')$ such that $x_1 \approx x_2$. By the continuity of $r(\cdot)$ and $r'(\cdot)$, we have:
    \begin{equation}
        r(x_1) - r'(x_1) = \Delta_k \quad \text{and} \quad r(x_2) - r'(x_2) = \Delta_{k'}.
    \end{equation}
    Assuming that $r(\cdot)$ and $r'(\cdot)$ share a Lipschitz constant $L$, we can bound the difference in parameter offsets between positions as:
    \begin{align}
    |\Delta_k - \Delta_{k'}| \leq 2L \cdot \|x_1 - x_2\|_2.
    \end{align} 
    This bound shows that when feature vectors are similar between positions, their parameter offsets must also be similar. As the feature overlap increases (i.e., as $\|x_1 - x_2\|_2 \to 0$), the difference in offsets approaches zero. Given that the graph $G$ is connected, there exists a path between any two positions $k$ and $k'$. Along this path, the parameter offsets between adjacent positions are approximately equal due to the continuity constraint. By transitivity across the connected graph, all position-dependent offsets $\Delta_k$ must converge to a single global offset $\Delta$ up to an approximation error that diminishes as feature overlap increases. After normalization (e.g., setting $\theta_1 = 0$), the model parameters are uniquely identified.
\end{proof}

\section{Practical pitfalls of overlapping features}
\label{appendix:practical-pitfalls}

The proof in Appendix~\ref{appendix:overlapping-features} makes two key assumptions for identifiability that can be difficult in practice: overlapping document features and a continuous relevance model. First, feature overlap decreases with increasing dimensionality~\cite{DAmour2021Overlap}. That is, as we add more query-document features, documents become less likely to be sufficiently close in feature space. Therefore, we should aim to use fewer query-document features, use dimensionality reduction methods, or introduce document swaps to guarantee overlapping support.

Secondly, the continuity assumption requires that small differences in feature space do not cause large, discontinuous jumps in relevance predictions. While neural networks are continuous, deep networks can produce large jumps in relevance prediction even for minor feature differences. Thus, when randomized data is unavailable, it is advisable to use shallow neural networks, regularization, or making parametric assumptions when possible to limit the expressiveness of the relevance model. Note that this advice may conflict with our discussion of model misspecification in Section~\ref{sec:logging_policy}.

\section{Experimental Setup}
\label{appendix:experimental-setup}

To evaluate our theoretical claims about model identifiability and logging policy effects, we implemented a comprehensive simulation framework. This controlled setting allows us to measure how well two-tower models can recover their underlying parameters under various conditions. Rather than solely relying on ranking metrics (such as nDCG), which can obscure underlying parameter estimation issues~\cite{Deffayet2023Robustness,Deffayet2023CMIP}, we directly compare inferred bias parameters against their ground-truth values to assess model identifiability. This provides clearer insights, as correctly estimated position bias parameters indicate properly identified relevance parameters due to their complementary relationship. By controlling logging policy strength, position randomization levels, and model specification, we will test each component of our theoretical analysis in isolation.

\paragraph{Datasets} The basis for our simulations are traditional learning-to-rank datasets with query-document features with an expert-annotated relevance label between 0-4. While we considered multiple datasets~\cite{Chapelle2011Yahoo,Dato2016Istella,Dato2022Istella22}, we will focus the discussion in the rest of the paper solely on MSLR30K~\cite{Qin2013MSLR} from the Bing search engine (31,531 queries and 136 dimensional feature vectors). We make this choice as our experiments consider very idealized simulations in which the choice of base dataset barely matters. We use the official training, validation, and testing splits. To increase computational efficiency, we truncate the number of documents per query to the 25 most relevant documents and drop all queries without a single relevant document ($\approx 3.1\%$ of queries). These steps purely aid the scalability of our simulation and all of our findings hold up on the full dataset. Next, we normalize each query-document feature provided using $\text{log1p}(x)=\ln(1 + |x|) \odot \text{sign}(x)$ following \citet{Qin2021AreNeuralRankers}.

\paragraph{Synthetic relevance} \label{sec:synthetic-relevance} LTR datasets come with relevance labels obtained by human experts~\cite{Qin2013MSLR,Chapelle2011Yahoo}. To isolate the effect of model mismatch, we additionally generate synthetic relevance labels that follow a known model class based on the provided query-document features. We generate a linear relevance label with Gaussian noise:
\begin{equation}\label{eq:linear-relevance}
    \hat{\gamma}_{q,d} = w^T x_{q,d} + \xi_{q,d}, \quad \xi_{q,d} \sim \mathcal{N}(0, \sigma^2),
\end{equation}
and a non-linear relevance label using a two-layer neural network with 16 neurons and tanh activations:
\begin{equation}\label{eq:non-linear-relevance}
    \hat{\gamma}_{q,d} = W_2^T \text{tanh}(W_1^T x_{q,d} + b_1) + b_2 + \xi_{q,d}, \quad \xi_{q,d} \sim \mathcal{N}(0, \sigma^2).
\end{equation}
In all our experiments we initialize the model weights randomly and add Gaussian noise with standard deviation $\sigma = 0.2$ to each label. Finally, we apply min-max scaling to ensure the resulting labels are in a comparable range to expert annotations, mapping the 5th to 95th percentiles of the synthetic labels to the range [0,4].

\paragraph{Logging policy} We isolate ranking performance and positional variability in our logging policy. To create a strong pointwise ranker, we train a relevance tower on the complete ground-truth relevance annotations of the training set with a mean-squared error loss. Note that we intentionally train our logging policy on the full training dataset and do not adhere to the common ULTR practice of training a weak ranker on 1\% of the dataset~\cite{Joachims2017IPW,Ai2018DLA,Jagerman2019ModelIntervene,Oosterhuis2020PolicyAware} as we need the strongest policy that our pointwise two-tower models could still capture. Our final logging policy score $s_{q,d}$ interpolates between our logging policy predictions and random noise:
\begin{equation}
\mbox{}\hspace*{-2mm}
    s_{q,d} = \operatorname{sign}(\alpha) \left(|\alpha| \cdot \hat{\gamma}_{q,d} + (1 -  |\alpha|) \cdot u_{q,d} \right), \ u_{q,d} \sim U(0, 4),
\hspace*{-2mm}\mbox{}    
\end{equation}
where $\alpha \in [-1, 1]$ is a hyperparameter to set the logging policy \emph{strength}. We simulate three levels of policy strength, using $\alpha = 1$ for the best possible ranking, $\alpha = 0$ for a random ranking, and $\alpha = -1$ for an inverse ranking that is worse than random by ranking documents from least to most relevant. Note that we sample noise only once per query-document pair and otherwise keep the ranking fixed across all user sessions for the same query. The noise interpolation setup follows \citet{Zhang2023Disentangling} but does not introduce additional biases as we sort by a trained logging policy and not directly by ground-truth expert labels.

Next, we transform our deterministic logging policy into a stochastic policy to introduce varying levels of positional variability into our simulation process, a crucial component to create overlap for our identifiability experiments. While various approaches exist to create stochastic policies, including the Gumbel-Max trick~\cite{Maddison2017Concrete,Bruch2019Softmax} or randomization schemes from the position bias literature \cite{Joachims2017IPW,Radlinski2006FairPairs}, we opt for an epsilon greedy strategy for simplicity and interpretability~\cite{Watkins1989EGreedy}. With this approach, we show a uniform random ranking with probability $\tau$, and otherwise rank deterministically according to our logging policy scores $s_{q,d}$. We refer to this probability $\tau$ as the \emph{temperature} of our logging policy, with $\tau = 0$ indicating a deterministic policy (showing the same ranking across all sessions) and $\tau = 1$ indicating a uniform random ranking per session.

\paragraph{Click simulation} We simulate user clicks with position bias following the two-tower model and leave an investigation of user model mismatch to future work. We define position bias logits as: $\hat{\theta}_k = - \ln(k)$ where $k$ is the rank of the current query-document pair. The final click probability is: $P(C = 1 \mid q, d, k) = \sigma(\hat{\theta}_k + (\hat{\gamma}_{q,d} - 2))$, with our final relevance logits being a zero-centered version of our relevance labels obtained earlier. We center the scores from their original range between [0, 4] to [-2, 2] to avoid click logits that might overly saturate the sigmoid. For all experiments, we simulate 1M user sessions for training and 500K sessions for validation and testing respectively. We repeat all simulations over three random seeds and plot the 95\% confidence interval across all figures.

\paragraph{Model} Lastly, we describe our model architecture. Our bias tower uses a single parameter per rank as the position bias logit $\theta_k$. Our relevance tower is either a single embedding parameter for each query-document pair $\gamma_{q,d}$, a linear model, or a two-layer feed-forward neural network (2 layers, 32 neurons, ELU activations). All models were optimized using AdamW~\cite{Loshchilov2019AdamW} with a learning rate of 0.003 and a weight decay of 0.01 up to 50 epochs, stopping early after three epochs of no improvement of the validation loss. The entire setup was implemented in Jax~\cite{Bradbury2018Jax}, Flax (NNX)~\cite{Heek2020Flax}, and Rax~\cite{Jagerman2022Rax}. Our code, data, and results are available here: \url{https://github.com/philipphager/two-tower-confounding}.

\section{Motivating example without expert policy}
\label{appendix:motivating-example}

We revisit our motivating example from the introduction and replace the expert policy (sorting by ground truth annotations) with a deep neural network trained on expert annotations. Removing this problematic source of model misspecification from our simulation greatly reduces the observed  logging policy impact.

\begin{figure}[H]
    \hspace*{-2em}
    \includegraphics[clip, trim=0mm -18mm 0mm 0mm, width=0.9\linewidth]{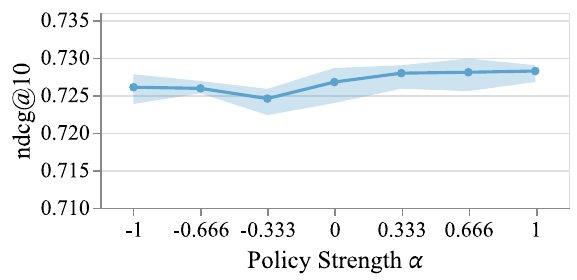}
    \vspace*{-1.4cm}
    \caption{Two-tower models trained on deterministic logging
    policies of varying strengths $\alpha$. The logging policy is a deep relevance tower trained on expert annotations instead of directly sorting by ground-truth relevance labels, as done in the introduction. This showcases the importance of model misspecification in ranking performance.}
    \label{fig:example-deep-policy}
\end{figure}

\section{Sample Weights for Misspecified Models}
\label{appendix:ips}

While the focus of this work is analytical, our findings on logging policy effects suggest a natural fix for training two-tower models. The priority should be to fix model misspecification, as this lowers the impact of unequal exposure through the logging policy. However, should that not be possible (e.g., because business rules that impacted our past rankings were not documented), we can reduce the impact of uneven document distributions by weighting each query-document pair inversely to its propensity of being displayed:
\begin{equation}
\mbox{}\hspace*{-2mm}
    \begin{split}
        \mathcal{\hat{L}_{\text{IPS}}}(\theta, \gamma: D) = - \frac{1}{N} \sum_{(q,d,k,c) \in \mathcal{D}} \frac{1}{\pi(d,k\mid q)} \bigl[c\log \sigma(\theta_k + \gamma_{q,d}) \\
        + (1-c)\log (1-\sigma(\theta_k + \gamma_{q,d}))\bigr].
    \end{split}
\hspace*{-2mm}\mbox{}    
\end{equation}
Note how this approach is different from traditional inverse propensity scoring (IPS) methods used in unbiased learning-to-rank~\cite{Joachims2017IPW,Oosterhuis2020PolicyAware} as we target different biases: existing pointwise IPS methods primarily correct for position bias~\cite{Bekker2019PointwiseIPS,Saito2020PointwiseIPS,Hager2023ClickModelIPS}. In contrast, this sample weight corrects for the uneven document distribution across ranks. The loss is essentially fitting a model against a version of the dataset in which all documents appeared equally in all positions, which is related to corrections used in the position bias estimation literature~\cite{Agarwal2019AllPairs,Fang2019InterventionHarvesting,Benedetto2023ContextualBias}, the policy-aware IPS estimator for selection bias~\cite{Oosterhuis2020PolicyAware,Li2018OffPolicyClickModels}, or the recent work by \citet{Luo2024UnbiasedPropensity}. Note that estimating this propensity is a challenge on its own. In this paper, we calculate the propensities by simply counting how many times a query-document pair was displayed in a given position. Lastly, Fig.~\ref{fig:model-misfit-ips} displays the results of applying the weighting scheme to the same three simulation setups used in our model mismatch experiments. We can see that the weighting scheme does not work under a deterministic logging policy $\tau = 0$, as documents have a probability of one to be displayed in their initial rank and a propensity of zero to be displayed in any other rank. We can see that in all other cases when we introduce increasing levels of variability into our policy, the models converge to identical position bias estimates, independent of the logging policy. Note that the estimated parameters are still slightly biased due to model misspecification, however, the bias amplification by the logging policy has been dampened.

\begin{figure*}[h!]
    \includegraphics[width=1\textwidth]{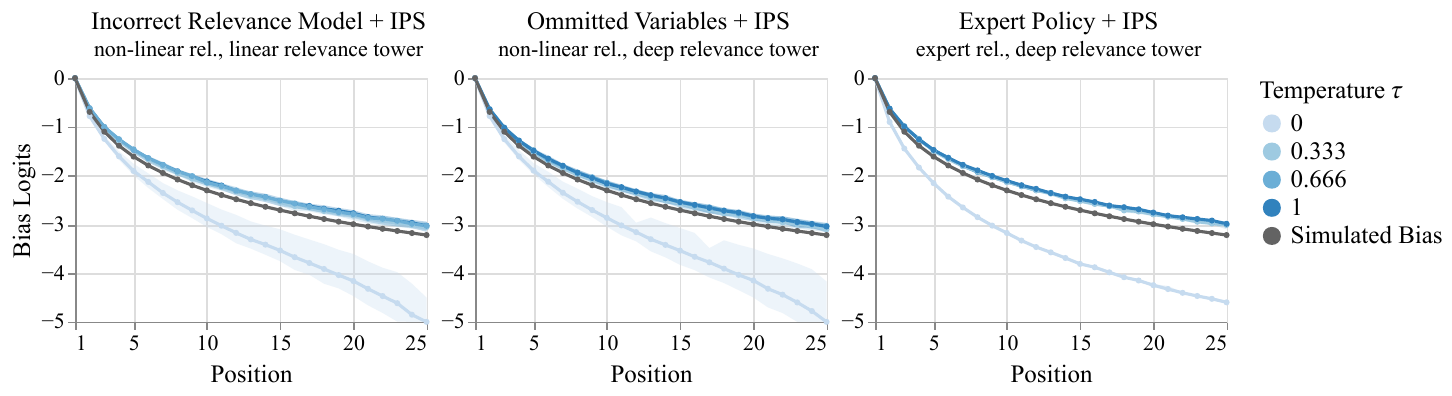}
    \caption{Reweighting samples during training inversely to the propensity of a policy placing an item in a given position. IPS requires that documents have a non-zero chance of occurring in other rank. Therefore, we see a helpful impact of IPS across various levels of randomization, but no improvement on a deterministic policy ($\tau = 0$). Note that while the logging policy impact is minimized, a bias remains in all three examples due to model misspecification.}
    \label{fig:model-misfit-ips}
\end{figure*}

\end{document}